\RequirePackage{fix-cm}
\documentclass[smallextended]{svjour3}       
\smartqed  
\usepackage{graphicx}
\usepackage{amsfonts}
\usepackage{amssymb,epic, graphicx}
\usepackage[cmex10]{amsmath}
\usepackage{array}
\usepackage{color}
\usepackage{mdwmath}
\usepackage{mdwtab}
\usepackage{eqparbox}
\usepackage[tight,footnotesize]{subfigure}
\usepackage{algorithm}
\usepackage{algorithmic}
\allowdisplaybreaks
\begin{document}

\title{A Class of Three-Weight Linear Codes and Their Complete Weight Enumerators
}

\titlerunning{A Class of Three-Weight Linear Codes}        

\author{Shudi Yang      \and
        Zheng-An Yao   \and~
        Chang-An Zhao
        }


\institute{S.D. Yang \at
              Department of Mathematics,
Sun Yat-sen University, Guangzhou 510275 and School of Mathematical
Sciences, Qufu Normal University, Shandong 273165, P.R.China \\
              \email{yangshd3@mail2.sysu.edu.cn}           
              \and
           Z.-A. Yao \at
               Department of Mathematics,
Sun Yat-sen University, Guangzhou 510275, P.R. China\\
              \email{mcsyao@mail.sysu.edu.cn}
              \and
               C.-A. Zhao \at
               Department of Mathematics,
Sun Yat-sen University, Guangzhou 510275, P.R. China\\
              \email{zhaochan3@mail.sysu.edu.cn}  }
            \date{Received: date / Accepted: date}

\maketitle

\begin{abstract}
Recently, linear codes constructed from defining sets have been investigated extensively and they have many applications. In this paper, for an odd prime $p$, we propose a class of $p$-ary linear codes by choosing a proper defining set. Their weight enumerators and complete weight enumerators are presented explicitly. The results show that they are linear codes with three weights and suitable for the constructions of authentication codes and secret sharing schemes.

\end{abstract}


\section{Introduction}\label{sec:intro}

Throughout this paper, let $p$ be an odd prime and $r=p^m$ for an integer $m\geq2$. Denote by $\mathbb{F}_r$ a
finite field with $r$ elements. An $[n, \kappa, \delta]$ linear code
$C$ over $\mathbb{F}_p$ is a $\kappa$-dimensional subspace of
$\mathbb{F}_p^n$ with minimum distance $\delta$
(see~\cite{ding2014differencesets,macwilliams1977theory}).

Let $A_i$ denote the number of codewords with Hamming weight
$i$ in a linear code $C$ of length $n$. The weight enumerator of $C$ is defined by
$A_0+A_1z+A_2z^2+\cdots+A_nz^n,$
where $A_0=1$. The sequence $(1,A_1,A_2,\cdots,A_n)$ is called the weight
distribution of the code $C$.

The complete weight enumerator of a code $C$ over $\mathbb{F}_p$ enumerates the codewords according to the number of symbols of each kind contained in each codeword. We recall the
 definition as that of~\cite{Blake1991}. Denote elements of the field by $\mathbb{F}_p=\{w_0,w_1,\cdots,w_{p-1}\}$, where $w_0=0$.
Also, let $\mathbb{F}_p^*$ denote $\mathbb{F}_p\backslash\{0\}$.
For a codeword $\mathsf{c}=(c_0,c_1,\cdots,c_{n-1})\in \mathbb{F}_p^n$, let $w[\mathsf{c}]$ be the
complete weight enumerator of $\mathsf{c}$, which is defined as
$$w[\mathsf{c}]=w_0^{k_0}w_1^{k_1}\cdots w_{p-1}^{k_{p-1}},$$
where $k_j$ is the number of components of $\mathsf{c}$ equal to $w_j$, $\sum_{j=0}^{p-1}k_j=n$.
The complete weight enumerator of the code $C$ is then
$$\mathrm{CWE}(C)=\sum_{\mathsf{c}\in C}w[\mathsf{c}].$$

The weight distributions of linear codes have been well studied in literature (see~\cite{ding2013hamming,dinh2015recent,feng2008weight,luo2008weight,sharma2012weight,vega2012weight,wang2012weight,yu2014weight,yuan2006weight,zheng2015weightseveral,Zhou2013fiveweight} and references therein). The information of the complete weight enumerators of linear codes is of vital use because they not only give the weight enumerators but also show the frequency of each symbol appearing in each codeword. Therefore, they have many applications. Blake and Kith investigated the complete weight enumerator of Reed-Solomon codes and showed that they could be helpful in soft decision decoding~\cite{Blake1991,kith1989complete}. In~\cite{helleseth2006}, the study of the monomial and quadratic bent functions was related to the complete weight enumerators of linear codes. It was illustrated by Ding $et~al.$~\cite{ding2007generic,Ding2005auth} that complete weight enumerators can be applied to the calculation of the deception probabilities of certain authentication codes. In~\cite{chu2006constantco,ding2008optimal,ding2006construction}, the authors studied the complete weight enumerators of some constant
composition codes and presented some families of optimal constant composition codes.

However, it is extremely difficult to evaluate the complete
weight enumerators of linear codes in general and there is little information on this topic in literature besides the above mentioned~\cite{Blake1991,chu2006constantco,ding2008optimal,ding2006construction,kith1989complete}.
Kuzmin and Nechaev investigated the
generalized Kerdock code and related linear codes over Galois rings and determined their complete weight enumerators in~\cite{kuzmin1999complete} and~\cite{kuzmin2001complete}. More recent progress on the complete weight enumerators of linear codes can be found in~\cite{AhnKaLi2016completegenelize,BaeLiYue2015complete,LiYang2015cwe,li2015complete,WangQiuyan2015complet}.
The results of~\cite{AhnKaLi2016completegenelize} and~\cite{BaeLiYue2015complete} can be
viewed as generalizations of~\cite{yang2015complete} and~\cite{ding2015twothree}, respectively.
In~\cite{LiYang2015cwe,li2015complete,WangQiuyan2015complet}, the authors treated the complete weight enumerators of some linear or cyclic codes by using exponential sums and Galois theory. It should be mentioned that Tang $et~al.$~\cite{tang2015linear} constructed linear codes with two or three weights from weakly regular bent functions. We shall generalize this construction to non-bent functions.

The authors of~\cite{ding2015twodesign,dingkelan2014binary,ding2015twothree} gave the generic construction of linear codes. Set $\bar{D}=\{d_1,d_2,\cdots,d_n\}\subseteq \mathbb{F}_{r}$, where $r=p^m$. Denote by $\mathrm{Tr}$ the absolute trace function. A linear code associated with $\bar{D}$ is defined by
\begin{equation*}\label{def:CD'}
    C_{\bar{D}}=\{(\mathrm{Tr}(ad_1),\mathrm{Tr}(ad_2),\cdots,\mathrm{Tr}(ad_n)):
       a\in \mathbb{F}_{r}\}.
\end{equation*}
Then $\bar{D}$ is called the defining set of this code $C_{\bar{D}}$.

Motivated by the above construction and the idea of~\cite{tang2015linear}, we define linear codes
$C_{D}$ and $C_{D_1}$ by
\begin{eqnarray}\label{def:CD1-1}
    &C_{D}&=\{(\mathrm{Tr}(ax^2))_{x\in D}: a\in \mathbb{F}_{r}\},\\
    &C_{D_1}&=\{(\mathrm{Tr}(ax^2))_{x\in D_1}: a\in \mathbb{F}_{r}\},\nonumber
\end{eqnarray}
where
\begin{eqnarray*}
     &D   &=\{x\in \mathbb{F}_{r}^*:\mathrm{Tr}(x)\in Sq\},\\
     &D_1&=\{x\in \mathbb{F}_{r}^*:\mathrm{Tr}(x)\in Nsq\},\nonumber
\end{eqnarray*}
are also called defining sets. Here $Sq$ and $Nsq$ denote the set of all square elements and non-square elements in $\mathbb{F}_{p}^*$, respectively. By definition, these codes have length $n=(p-1)p^{m-1}/2$ and dimension at most $m$. Further, we will demonstrate that $C_D$ is equal to $C_{D_1}$. Actually, for a fixed $b\in Nsq$, there exists a mapping $\phi_b$ such that
\begin{eqnarray*}
     \phi_b : &D &\rightarrow D_1\\
              &x& \mapsto bx
\end{eqnarray*}
which implies that $\mathrm{Tr}(a(\phi_b(x))^2)=\mathrm{Tr}(ab^2x^2)$
for all $x\in D$ and $a\in \mathbb{F}_{r}$.
As $a$ runs through $\mathbb{F}_{r}$, so does $ab^2$. This means they have the same codewords. Hence, we only describe all the information of $C_{D}$. In this paper, the complete weight enumerator of $C_{D}$ is investigated by employing
exponential sums and Gauss periods. This gives its weight enumerator immediately. As it turns out, this code is a three-weight linear code which will be of special interest in authentication codes~\cite{Ding2005auth} and
secret sharing schemes~\cite{carlet2005linear}.

The remainder of this paper is organized as follows. In Section
\ref{sec:main results}, we describe the main results of this paper, additionally we give some examples. Section~\ref{sec:pf} briefly recalls some definitions and results on Gauss periods and Gauss sums, then proves the main results.
Finally, Section~\ref{sec:conclusion} is devoted to conclusions.

\section{Main results}\label{sec:main results}

 In this section, we only introduce the complete weight enumerator and weight enumerator of $C_{D}$ described in~\eqref{def:CD1-1}. The main results of this paper are presented below, whose proofs will be given in Section~\ref{sec:pf}.

First of all, we establish the complete weight enumerator of $C_{D}$ in the following three theorems, after which, we give some examples to illustrate these results.

\begin{theorem}\label{thcwe:CD1,-1}
Let $p\equiv3\mod 4$ and $\rho, z$ be elements in $\mathbb{F}_{p}$. Then the code $C_{D}$ defined by~\eqref{def:CD1-1} is a $[\frac{p-1}{2}p^{m-1},m]$ three-weight linear code and we have the following assertions.\\
$(i)$ If $m$ is even, then the
complete weight enumerator of $C_D$ is given by
\begin{eqnarray*}
&&w_0^{\frac{p-1}{2}p^{m-1}}+(p^{m-1}-1)\prod_{\rho\in \mathbb{F}_{p}}w_{\rho}^{\frac{p-1}{2}p^{m-2}}\\
&&+\frac{p-1}{4}(p^{m-1}\!+\!p^{\frac{m-2}{2}})w_0^{\frac{p-1}{2}(p^{m-2}-p^{\frac{m-2}{2}})}
\sum_{i\in \{1,-1\}}\prod_{\left(\frac{\rho}{p}\right)=i}w_{\rho}^{\frac{p-1}{2}p^{m-2}}
\prod_{\left(\frac{z}{p}\right)=-i}w_{z}^{A_1}\\
&&+\frac{p-1}{4}(p^{m-1}\!-\!p^{\frac{m-2}{2}})w_0^{\frac{p-1}{2}(p^{m-2}+p^{\frac{m-2}{2}})}
\sum_{i \in \{1,-1\}}\prod_{\left(\frac{\rho}{p}\right)=i}w_{\rho}^{\frac{p-1}{2}p^{m-2}}
\prod_{\left(\frac{z}{p}\right)=-i}w_{z}^{A_{-1}},
\end{eqnarray*}
where, for $\varepsilon \in \{1,-1\}$,
\begin{eqnarray*}
A_\varepsilon=\frac{p-1}{2}p^{m-2}+\varepsilon p^{\frac{m-2}{2}}.
\end{eqnarray*}\\
$(ii)$ If $m$ is odd, then the
complete weight enumerator of $C_D$ is given by
\begin{eqnarray*}
&&w_0^{\frac{p-1}{2}p^{m-1}}+(p^{m-1}-1)\prod_{\rho\in \mathbb{F}_{p}}w_{\rho}^{\frac{p-1}{2}p^{m-2}}\\
&&+\frac{p-1}{4}(p^{m-1}\!+\!p^{\frac{m-1}{2}})w_0^{\frac{p-1}{2}(p^{m-2}-p^{\frac{m-3}{2}})}
\sum_{i\in\{1,-1\}}\prod_{\left(\frac{\rho}{p}\right)=i}w_{\rho}^{A_1}
\prod_{\left(\frac{z}{p}\right)=-i}w_{z}^{B_1}\\
&&+\frac{p-1}{4}(p^{m-1}\!-\!p^{\frac{m-1}{2}})w_0^{\frac{p-1}{2}(p^{m-2}+p^{\frac{m-3}{2}})}
\sum_{i\in\{1,-1\}}\prod_{\left(\frac{\rho}{p}\right)=i}w_{\rho}^{A_{-1}}
\prod_{\left(\frac{z}{p}\right)=-i}w_{z}^{B_{-1}},
\end{eqnarray*}
where, for $\varepsilon \in \{1,-1\}$,
\begin{eqnarray*}
A_\varepsilon &=&\frac{p-1}{2}(p^{m-2}-\varepsilon p^{\frac{m-3}{2}}),\\
B_\varepsilon &=& \frac{p-1}{2}p^{m-2}+\varepsilon\frac{p+1}{2}p^{\frac{m-3}{2}}.
\end{eqnarray*}
\end{theorem}

\begin{example}
(i) Let $(p,m)=(3,5)$. Then by Theorem~\ref{thcwe:CD1,-1}, the code $C_{D}$ has
parameters $[81, 5, 51]$ and complete weight enumerator
\begin{eqnarray*}
w_0^{81}&+&36w_0^{30}w_1^{30}w_2^{21}+36w_0^{30}w_1^{21}w_2^{30}+80w_0^{27}w_1^{27}w_2^{27}\\
&+&45w_0^{24}w_1^{33}w_2^{24}+45w_0^{24}w_1^{24}w_2^{33},
\end{eqnarray*}
which is verified by Magma program. This is a three-weight linear code.

(ii) Let $(p,m)=(7,2)$. Then by Theorem~\ref{thcwe:CD1,-1}, the code $C_{D}$ is
a $[21, 2, 15]$ three-weight linear code with
complete weight enumerator
\begin{eqnarray*}
w_0^{21}&+&6(w_0w_1w_2 w_3w_4 w_5w_6)^{3}\\
&+&9w_0^{6}(w_1w_2w_4)^{3}(w_3w_5w_6)^{2}+9w_0^{6}(w_1w_2w_4)^{2}(w_3w_5w_6)^{3}\\
&+&12(w_1w_2w_4)^{4}(w_3w_5w_6)^{3}+12(w_1w_2w_4)^{3}(w_3w_5w_6)^{4},
\end{eqnarray*}
which is confirmed by Magma program.

\end{example}

Let $p\equiv1\mod 4$. For $i=0,1,2,3$, we denote the cyclotomic classes of order 4 in $\mathbb{F}_{p}$ by $C_i^{(4,p)}$, which is simplified as $C_i$ in the sequel.
\begin{theorem}\label{thcwe:CD1,-1 p=1mod4}
Let $p\equiv1\mod 4$ and $m$ be odd. Then the code $C_{D}$ of~\eqref{def:CD1-1} is a $[\frac{p-1}{2}p^{m-1},m]$ three-weight linear code with
complete weight enumerator
\begin{eqnarray*}
&&w_0^{\frac{p-1}{2}p^{m-1}}+(p^{m-1}-1)\prod_{\rho\in \mathbb{F}_{p}}w_{\rho}^{\frac{p-1}{2}p^{m-2}}\\
&&+\frac{p-1}{8}(p^{m-1}+ p^{\frac{m-1}{2}})\sum_{i=0}^3 w_0^{\frac{p-1}{2}(p^{m-2}-p^{\frac{m-3}{2}})}
\prod_{\rho\in C_i}w_{\rho}^{A_1} \prod_{z\in \mathbb{F}_{p}^*\setminus C_i}w_{z}^{B_1}\\
&&+\frac{p-1}{8}(p^{m-1}- p^{\frac{m-1}{2}})\sum_{i=0}^3 w_0^{\frac{p-1}{2}(p^{m-2}+p^{\frac{m-3}{2}})}
\prod_{\rho\in C_i}w_{\rho}^{A_{-1}} \prod_{z\in \mathbb{F}_{p}^*\setminus C_i}w_{z}^{B_{-1}}
,
\end{eqnarray*}
where, for $\varepsilon \in \{1,-1\}$,
\begin{eqnarray*}
A_\varepsilon&=&
    \frac{p-1}{2}p^{m-2}+\frac{\varepsilon }{2}(3 p^{\frac{m-1}{2}}+p^{\frac{m-3}{2}}),\\
B_\varepsilon &=&
 \frac{p-1}{2}p^{m-2}- \frac{\varepsilon }{2}( p^{\frac{m-1}{2}}-p^{\frac{m-3}{2}}).
\end{eqnarray*}

\end{theorem}

\begin{example}
Let $(p,m)=(5,3)$. Then by Theorem~\ref{thcwe:CD1,-1 p=1mod4}, the code $C_{D}$ is a three-weight linear code with parameters $[50, 3, 38]$ and complete weight enumerator
\begin{eqnarray*}
w_0^{50}&+&10(w_0w_1w_2w_3)^{12}w_4^{2}+10(w_0w_1w_2w_4)^{12}w_3^{2}+10(w_0w_1w_3w_4)^{12}w_2^{2}\\
&+&10(w_0w_2w_3w_4)^{12}w_1^{2}+24(w_0w_1w_2w_3w_4)^{10}+15(w_0w_1w_2w_3)^{8}w_4^{18}\\
&+&15(w_0w_1w_2w_4)^{8}w_3^{18}+15(w_0w_1w_3w_4)^{8}w_2^{18}+15(w_0w_2w_3w_4)^{8}w_1^{18}.
\end{eqnarray*}
These results can be checked by Magma program.

\end{example}

\begin{theorem}\label{thcwe:CD1,-1 p=1mod4 2}
Let $p\equiv1\mod 4$ and $m$ be even. Let $s$ and $t$ be defined by $p=s^2+t^2$, $s\equiv1\mod 4$. Then the code $C_{D}$ of~\eqref{def:CD1-1} is a $[\frac{p-1}{2}p^{m-1},m]$ three-weight linear code with complete weight enumerator given by
\begin{eqnarray*}
&&w_0^{\frac{p-1}{2}p^{m-1}}+(p^{m-1}-1)\prod_{\rho\in \mathbb{F}_{p}}w_{\rho}^{\frac{p-1}{2}p^{m-2}}\\
&&+\frac{p-1}{8}(p^{m-1}+ p^{\frac{m-2}{2}})\sum_{i=0}^3 w_0^{K_1}
\prod_{\rho_0\in C_i}w_{\rho_0}^{L_1} \prod_{\rho_1\in C_{i+1}}w_{\rho_1}^{R_1}
\prod_{\rho_2\in C_{i+2}}w_{\rho_2}^{S_1} \prod_{\rho_3\in C_{i+3}}w_{\rho_3}^{T_1}\\
&&+\frac{p\!-\!1}{8}(p^{m-1}\!-\! p^{\frac{m-2}{2}})\sum_{i=0}^3 \! w_0^{K_{\!-1\!}}
\prod_{\rho_0\in C_i}\!w_{\rho_0}^{L_{\!-1\!}} \prod_{\rho_1\in C_{i+1}}\!w_{\rho_1}^{R_{\!-1\!}}
\prod_{\rho_2\in C_{i+2}}\!w_{\rho_2}^{S_{\!-1\!}} \prod_{\rho_3\in C_{i+3}}\!w_{\rho_3}^{T_{\!-1\!}},
\end{eqnarray*}
where, for $\varepsilon \in \{1,-1\}$,
\begin{eqnarray*}
K_\varepsilon &=&\frac{p-1}{2}(p^{m-2}-\varepsilon p^{\frac{m-2}{2}}),\\
L_\varepsilon &=& \frac{p-1}{2}p^{m-2}+\varepsilon p^{\frac{m-2}{2}}(1+s),\\
R_\varepsilon &=& \frac{p-1}{2}p^{m-2}- \varepsilon p^{\frac{m-2}{2}}t,\\
S_\varepsilon &=& \frac{p-1}{2}p^{m-2}+ \varepsilon p^{\frac{m-2}{2}}(1-s),\\
T_\varepsilon &=& \frac{p-1}{2}p^{m-2}+ \varepsilon p^{\frac{m-2}{2}}t.
\end{eqnarray*}

\end{theorem}

\begin{example}
Let $(p,m)=(5,4)$. Then by Theorem~\ref{thcwe:CD1,-1 p=1mod4 2}, the code $C_{D}$ has parameters $[250, 4, 190]$ and complete weight enumerator
\begin{eqnarray*}
w_0^{250}&+&60w_0^{60}w_1^{60}w_2^{40}w_3^{50}w_4^{40}+60w_0^{60}w_1^{50}w_2^{60}w_3^{40}w_4^{40}+60w_0^{60}w_1^{40}w_2^{50}w_3^{40}w_4^{60}\\
&+&60w_0^{60}w_1^{40}w_2^{40}w_3^{60}w_4^{50}+124(w_0w_1w_2w_3w_4)^{50}+65w_0^{40}w_1^{60}w_2^{60}w_3^{40}w_4^{50}\\
&+&65w_0^{40}w_1^{60}w_2^{50}w_3^{60}w_4^{40}+65w_0^{40}w_1^{50}w_2^{40}w_3^{60}w_4^{60}+65w_0^{40}w_1^{40}w_2^{60}w_3^{50}w_4^{60},
\end{eqnarray*}
which is verified by Magma program. This is a three-weight linear code.

\end{example}

The following corollary gives the weight enumerator of $C_{D}$, which follows immediately from its complete weight enumerator.
\begin{corollary}\label{wt:CD1,-1}

 The code $C_{D}$ of~\eqref{def:CD1-1} has weight distribution given in Table~\ref{wt:m even}
if $m$ is even and Table~\ref{wt:m odd} if $m$ is odd.

\begin{table}[htbp]
\tabcolsep 2mm \caption{The weight distribution of $C_{D}$
if $m$ is even}\label{wt:m even}
\begin{center}\begin{tabular}{ll}
\hline\noalign{\smallskip}
  Weight $i$   & Frequency  $A_i$            \\
    \noalign{\smallskip}\hline\noalign{\smallskip}
  $\frac{(p-1)^2}{2}p^{m-2}$ & $p^{m-1}-1$     \\
   $\frac{p-1}{2}\left((p-1)p^{m-2}+p^{\frac{m-2}{2}}\right)$
            &$\frac{p-1}{2}(p^{m-1}+p^{\frac{m-2}{2}}) $   \\
   $\frac{p-1}{2}\left((p-1)p^{m-2}-p^{\frac{m-2}{2}}\right)$
            &$\frac{p-1}{2}(p^{m-1}-p^{\frac{m-2}{2}}) $   \\
   0         &  1        \\
   \noalign{\smallskip}\hline
  \end{tabular}
\end{center}

\end{table}

\begin{table}[htbp]
\tabcolsep 2mm \caption{The weight distribution of $C_{D}$
if $m$ is odd}\label{wt:m odd}
\begin{center}\begin{tabular}{ll}
\hline\noalign{\smallskip}
  Weight $i$   & Frequency  $A_i$            \\
    \noalign{\smallskip}\hline\noalign{\smallskip}
  $\frac{(p-1)^2}{2}p^{m-2}$ & $p^{m-1}-1$     \\
  $\frac{p-1}{2}\left((p-1)p^{m-2}+p^{\frac{m-3}{2}}\right)$
 &$\frac{p-1}{2}(p^{m-1}+p^{\frac{m-1}{2}}) $   \\
 $\frac{p-1}{2}\left((p-1)p^{m-2}-p^{\frac{m-3}{2}}\right)$
 &$\frac{p-1}{2}(p^{m-1}-p^{\frac{m-1}{2}}) $   \\
   0         &  1        \\
   \noalign{\smallskip}\hline
  \end{tabular}
\end{center}

\end{table}

\end{corollary}

%
%

From Tables~\ref{wt:m even} and~\ref{wt:m odd}, we observe that the weights of $C_{D}$
have a common divisor $(p-1)/2$. This implies that it can be punctured into a shorter code as follows.

Note that for any $a\in Sq$ and $x\in \mathbb{F}_{r}^* $, $\mathrm{Tr}(x)=0$ if and only if $\mathrm{Tr}(ax)=a\mathrm{Tr}(x)=0$. Then the defining set $D$ can be expressed as
\begin{eqnarray*}
D= Sq \tilde{D} = \{a\tilde{d}:a\in Sq, \tilde{d}\in \tilde{D} \},
\end{eqnarray*}
such that $\tilde{d}_i / \tilde{d}_j \not \in Sq $ for every pair of distinct elements $\tilde{d}_i $,
 $ \tilde{d}_j$ in $\tilde{D}$. Hence, the corresponding linear code $C_{\tilde{D}}$ is the punctured version of $C_{D}$. The following corollary states the parameters and weight distribution of $C_{\tilde{D}}$, which directly follows from Corollary~\ref{wt:CD1,-1}.

 \begin{corollary}\label{wt2:CD1,-1} The code $C_{\tilde{D}}$ is a $[p^{m-1},m]$ three-weight linear codes with weight distribution given in Table~\ref{wt2:m even}
if $m$ is even and Table~\ref{wt2:m odd} if $m$ is odd.
\begin{table}[htbp]
\tabcolsep 2mm \caption{The weight distribution of $C_{\tilde{D}}$
if $m$ is even}\label{wt2:m even}
\begin{center}\begin{tabular}{ll}
\hline\noalign{\smallskip}
  Weight $i$   & Frequency  $A_i$            \\
    \noalign{\smallskip}\hline\noalign{\smallskip}
  $(p-1)p^{m-2}$ & $p^{m-1}-1$     \\
   $(p-1)p^{m-2}+p^{\frac{m-2}{2}}$
            &$\frac{p-1}{2}( p^{m-1}+p^{\frac{m-2}{2}}) $   \\
   $(p-1)p^{m-2}-p^{\frac{m-2}{2}}$
            &$\frac{p-1}{2}(p^{m-1}-p^{\frac{m-2}{2}}) $   \\
   0         &  1        \\
   \noalign{\smallskip}\hline
  \end{tabular}
\end{center}

\end{table}

\begin{table}[htbp]
\tabcolsep 2mm \caption{The weight distribution of $C_{\tilde{D}}$
if $m$ is odd}\label{wt2:m odd}
\begin{center}\begin{tabular}{ll}
\hline\noalign{\smallskip}
  Weight $i$   & Frequency  $A_i$            \\
    \noalign{\smallskip}\hline\noalign{\smallskip}
  $(p-1)p^{m-2}$ & $p^{m-1}-1$     \\
  $(p-1)p^{m-2}+p^{\frac{m-3}{2}}$
 &$\frac{p-1}{2}(p^{m-1}+p^{\frac{m-1}{2}} )$   \\
 $(p-1)p^{m-2}-p^{\frac{m-3}{2}}$
 &$\frac{p-1}{2}(p^{m-1}-p^{\frac{m-1}{2}} )$   \\
   0         &  1        \\
   \noalign{\smallskip}\hline
  \end{tabular}
\end{center}

\end{table}

\end{corollary}

\begin{example}
 Let $(p,m)=(5,3)$. Then the code $C_{\tilde{D}}$ in Corollary~\ref{wt2:CD1,-1} has parameters $[25, 3, 19]$ and weight enumerator
\begin{eqnarray*}
1+40z^{19}+24z^{20}+60z^{21}.
\end{eqnarray*}
The code is almost optimal in the sense that the best known code over $\mathbb{F}_{5}$ of length 25 and dimension 3
has minimum distance 20 according to Markus Grassl's table (see http://www.codetables.de/).

\end{example}

\section{The proofs of the main results}\label{sec:pf}

\subsection{Auxiliary results}\label{subsec:math tool}
In order to prove Theorems~\ref{thcwe:CD1,-1},~\ref{thcwe:CD1,-1 p=1mod4} and~\ref{thcwe:CD1,-1 p=1mod4 2} proposed in Section~\ref{sec:main results}, we will use several results which are depicted and proved in the sequel. We start with cyclotomic classes and group characters.

Recall that $r=p^m$. Let $\alpha$ be a fixed primitive element of $\mathbb{F}_r$ and $r-1=sN$, where $s$, $N$ are two integers with $s>1$ and $N>1$. Define $C_i^{(N,r)}=\alpha
^i\langle\alpha^N\rangle$ for $i=0,1,\cdots,N-1$, where
$\langle\alpha^N\rangle$ denotes the subgroup of $\mathbb{F}_r^*$
generated by $ \alpha^N$. The cosets $C_i^{(N,r)}$ are called the \emph{cyclotomic classes} of order $N$
in $\mathbb{F}_r$.

For each $b\in\mathbb{F}_r$, let $\chi_b$ be an additive character of $\mathbb{F}_r$, which is defined by
\begin{equation*}
\chi_b(x)=\zeta_p^{\text{Tr}(bx)} ~~ \mathrm{for~~ all~~ }x\in\mathbb{F}_r.
\end{equation*}
Here $\zeta_p=\exp\left(\frac{2\pi\sqrt{-1}}{p}\right)$
and $\text{Tr}$ is the absolute trace function. Especially when $b=1$, $\chi_1$ is called the canonical additive character of $\mathbb{F}_r$.
The orthogonal property of additive characters $\chi$, which can
be easily checked, is given by
\begin{equation}\label{eq:ortho}
\sum_{x\in \mathbb{F}_r}\chi(ax)
=\left\{\begin{array}{lll}
r~~~~&&\mbox{if}~~a=0, \\
0~~~~&&\mbox{if}~~a\in\mathbb{F}_r^*.
\end{array}
\right.
\end{equation}

The Gauss
periods of order $N$ are defined by
\begin{equation*}
\eta_{i}^{(N,r)}=\sum_{x\in C_i^{(N,r)}}\chi_1(x), ~~i=0,1,\cdots,N-1.
\end{equation*}



Let $\lambda$ be a multiplicative and $\chi$ an additive character of $\mathbb{F}_r$. Then the Gauss sum $G(\lambda,\chi)$ is defined by
\begin{eqnarray*}
G(\lambda,\chi)=\sum_{x\in\mathbb{F}_r^*}\lambda(x)\chi(x).
\end{eqnarray*}

Let $\eta$ denote the quadratic character of $\mathbb{F}_r$. The associated
Gauss sum $G(\eta, \chi_1 )$ over $\mathbb{F}_{r}$ is denoted by $G(\eta)$. And
the Gauss sum $G(\bar{\eta},\bar{\chi}_1)$ over $\mathbb{F}_{p}$ is denoted by $G(\bar{\eta})$,
where $\bar{\eta}$ and $\bar{\chi}_1$ are the quadratic character and canonical additive character of $\mathbb{F}_{p}$, respectively.

For each $y \in \mathbb{F}_{p}^*$, we have
$\eta(y) = 1$ if $m\geq 2$ is even, and otherwise $\eta(y) = \bar{\eta}(y)$. Moreover, it is well known that $G(\eta)=(-1)^{m-1}\sqrt{p^*}^m$ and $G(\bar{\eta})=\sqrt{p^*}$, where $p^*=\left(\frac{-1}{p}\right)p=(-1)^{\frac{p-1}{2}}p$. See~\cite{ding2015twothree,lidl1983finite} for more information.


The following lemmas will be useful in the sequel.
\begin{lemma}(See Theorem 5.30 of~\cite{lidl1983finite})\label{lm:expo sum k}
Let $\chi$ be a nontrivial additive character of $\mathbb{F}_{r}$, $k\in \mathbb{N}$, and $\lambda$ a multiplicative character of $\mathbb{F}_{r}$ of order $d=\mathrm{gcd}(k,r-1)$. Then
\begin{eqnarray*}\label{eq:expo sum}
\sum_{x\in
\mathbb{F}_{r}}\chi(ax^k+b)=\chi(b)\sum_{j=1}^{d-1}\tilde{\lambda}^j(a)G({\lambda}^j,\chi)
\end{eqnarray*}
for any $a,b\in \mathbb{F}_{r}$ with $a\neq 0$.
Here $\tilde{\lambda}$ denotes the conjugate character of $\lambda$.
\end{lemma}

For $\rho\in \mathbb{F}_{p}^*$ and $a\in \mathbb{F}_{r}$, in order to study the complete weight enumerator, we define
\begin{eqnarray*}
  &N_0(\rho)&=\#\{x\in\mathbb{F}_{r}:\mathrm{Tr}(x)=0, \mathrm{Tr}(ax^2)=\rho\},\\
  &N(\rho)&=\#\{x\in\mathbb{F}_{r}:\mathrm{Tr}(x)\in Sq, \mathrm{Tr}(ax^2)=\rho\},\\
  &N_1(\rho)&=\#\{x\in\mathbb{F}_{r}:\mathrm{Tr}(x)\in Nsq, \mathrm{Tr}(ax^2)=\rho\}.
 \end{eqnarray*}
 The values of $N(\rho)$, $N_0(\rho)$
 and $N_1(\rho)$, which depend mainly on the choice of $a$, are given in the following two lemmas.

\begin{lemma}(\cite{yang2015complete})\label{lem:N0 rho}
Let $a\in \mathbb{F}_{r}^* $ and $\rho\in \mathbb{F}_{p}^*$. Then
\begin{eqnarray*}
  N_0(\rho)=\left\{\begin{array}{lll}
    &p^{m-2}+(-1)^{\frac{p-1}{2}\frac{m-1}{2}}\eta(a)\bar{\eta}(\rho)p^{\frac{m-1}{2}}
    ~~~~~~~~~~if~m~odd ,\mathrm{Tr}(a^{-1}\!)=0,\\
    &p^{m-2}-(-1)^{\frac{p-1}{2}\frac{m-1}{2}}\eta(a)\bar{\eta}(\mathrm{Tr}(a^{-1}\!))p^{\frac{m-3}{2}} ~~if~m~odd ,\mathrm{Tr}(a^{-1})\neq0,\\    &p^{m-2}+(-1)^{\frac{p-1}{2}\frac{m}{2}}\eta(a)p^{\frac{m-2}{2}}
    ~~~~~~~~~~~~~~~~~if~m~even ,\mathrm{Tr}(a^{-1}\!)=0,\\
    &p^{m-2}-(-1)^{\frac{p-1}{2}\frac{m-2}{2}}\eta(a)\bar{\eta}(\rho\mathrm{Tr}(a^{-1}\!))
   p^{\frac{m-2}{2}} \\
    &~~~~~~~~~~~~~~~~~~~~~~~~~~~~~~~~~~~~~~~~~~~~~~~~~~~if~m~even ,\mathrm{Tr}(a^{-1}\!)\neq0.
\end{array} \right.
 \end{eqnarray*}
\end{lemma}

%

 \begin{lemma}\label{lem:sum2}
Let $a\in \mathbb{F}_{r}^* $ and $\rho\in \mathbb{F}_p^* $. Then we have the following assertion.
\begin{eqnarray*}
 &&N(\rho)+N_1(\rho)\\
 &&=\left\{\begin{array}{lll}
   p^{m-1}-p^{m-2} ~~~~~~~~~~~~~~~~~~~~~~~~~~~~~~~~if~m~even,\mathrm{Tr}(a^{-1}\!)=0,\\
   p^{m-1}-p^{m-2} ~~~~~~~~~~~~~~~~~~~~~~~~~~~~~~~~if~m~odd,~\mathrm{Tr}(a^{-1}\!)=0,\\
   p^{m-1}-p^{m-2}+\eta(a)(-1)^{\frac{p-1}{2}\frac{m}{2}}p^{\frac{m-2}{2}}
   \left(1+\bar{\eta}(-\rho\mathrm{Tr}(a^{-1}\!))\right) \\ ~~~~~~~~~~~~~~~~~~~~~~~~~~~~~~~~~~~~~~~~~~~~~~~~if~m~even,\mathrm{Tr}(a^{-1}\!)\neq0,\\
   p^{m-1}-p^{m-2}+\eta(a)(-1)^{\frac{p-1}{2}\frac{m-1}{2}}p^{\frac{m-3}{2}}
   \left(\bar{\eta}(\rho)p+\bar{\eta}(\mathrm{Tr}(a^{-1}\!))\right) \\    ~~~~~~~~~~~~~~~~~~~~~~~~~~~~~~~~~~~~~~~~~~~~~~~~if~m~odd,~\mathrm{Tr}(a^{-1}\!)\neq0.
\end{array} \right.
 \end{eqnarray*}
\end{lemma}
\begin{proof}

Note that
\begin{eqnarray*}
N_0(\rho)+N(\rho)+N_1(\rho)=\#\{x\in\mathbb{F}_{r}: \mathrm{Tr}(ax^2)=\rho\},
 \end{eqnarray*}
 where $\rho\in \mathbb{F}_p^* $. This leads to
\begin{eqnarray*}
N_0(\rho)+N(\rho)+N_1(\rho)
=p^{m-1}+p^{-1}\sum_{z\in\mathbb{F}_{p}^*}\zeta_p^{-z\rho}\sum_{x\in\mathbb{F}_{r}}\zeta_p^{z\mathrm{Tr}(ax^2)}.
 \end{eqnarray*}
Applying Theorem 5.33 of~\cite{lidl1983finite}, we can deduce that
\begin{eqnarray*}
 \sum_{z\in\mathbb{F}_{p}^*}\sum_{x\in\mathbb{F}_{r}}\zeta_p^{z\mathrm{Tr}(ax^2)-z\rho}
  =\left\{\begin{array}{lll}
   \eta(a)(-1)^{\frac{p-1}{2}\frac{m}{2}}p^{\frac{m}{2}}         &&~\mathrm{if}~m~\mathrm{even},\\
   \eta(a)\bar{\eta}(\rho)(-1)^{\frac{p-1}{2}\frac{m-1}{2}}p^{\frac{m+1}{2}}    &&~\mathrm{if}~m~\mathrm{odd}.
\end{array} \right.
 \end{eqnarray*}
The desired conclusion then follows from Lemma~\ref{lem:N0 rho}. \hfill\space$\qed$
\end{proof}

%
%

%
%

The following two lemmas will help us to determine the frequency of each complete weight in $C_D$.
\begin{lemma}(\cite{yang2015complete})\label{lem:nij1}
For any $a\in \mathbb{F}_{r}^* $, let
\begin{eqnarray}\label{def:nij}
    n_{i,j}&=&\#\{a\in \mathbb{F}_{r}^*:\eta(a)=i ,~ \bar{\eta}(\mathrm{Tr}(a^{-1}))=j \},
    ~~i,j\in\{1,-1\}.
\end{eqnarray}\\
$(i)$ If $m$ is even, then we have
\begin{eqnarray*}
     n_{1,1}=   n_{1,-1}=\frac{p-1}{4}\left(p^{m-1}+(-1)^{\frac{p-1}{2}\frac{m}{2}}p^{\frac{m-2}{2}}\right).
 \end{eqnarray*}\\
$(ii)$ If $m$ is odd, then we have
\begin{eqnarray*}
  \left\{\begin{array}{lll}
   n_{1,1}&=&\frac{p-1}{4}\left(p^{m-1}+(-1)^{\frac{p-1}{2}\frac{m-1}{2}}p^{\frac{m-1}{2}}\right),\\
   n_{1,-1}&=&\frac{p-1}{4}\left(p^{m-1}-(-1)^{\frac{p-1}{2}\frac{m-1}{2}}p^{\frac{m-1}{2}}\right).\\
\end{array} \right.
 \end{eqnarray*}
\end{lemma}

\begin{lemma}\label{lem:nij2}
For any $a\in \mathbb{F}_{r}^* $, let $n_{i,j}$ be defined by~\eqref{def:nij}.\\
$(i)$ If $m$ is even, then we have
\begin{eqnarray*}
     n_{-1,1}=   n_{-1,-1}=\frac{p-1}{4}\left(p^{m-1}-(-1)^{\frac{p-1}{2}\frac{m}{2}}p^{\frac{m-2}{2}}\right).
 \end{eqnarray*}\\
$(ii)$ If $m$ is odd, then we have
\begin{eqnarray*}
  \left\{\begin{array}{lll}
   n_{-1,1}&=&\frac{p-1}{4}\left(p^{m-1}-(-1)^{\frac{p-1}{2}\frac{m-1}{2}}p^{\frac{m-1}{2}}\right),\\
   n_{-1,-1}&=&\frac{p-1}{4}\left(p^{m-1}+(-1)^{\frac{p-1}{2}\frac{m-1}{2}}p^{\frac{m-1}{2}}\right).\\
\end{array} \right.
 \end{eqnarray*}
\end{lemma}

\begin{proof}
We point out that
\begin{eqnarray*}
 n_{1,j}+n_{-1,j}=\#\{a\in \mathbb{F}_{r}^*:\bar{\eta}(\mathrm{Tr}(a^{-1}))=j\}=\frac{p-1}{2}p^{m-1},
 \end{eqnarray*}
with $j\in\{1,-1\}$.

The desired conclusion then follows from Lemma~\ref{lem:nij1}.
\hfill\space$\qed$
\end{proof}

Consider $p\equiv1\mod 4$. Recall that $\eta_{i}^{(4,p)}=\sum_{x\in C_i^{(4,p)}} \zeta_p^{x}$, where $C_i^{(4,p)}=\beta^i\langle\beta^4\rangle$ for $i=0,1,2,3$, and $\beta$ is a primitive element of $\mathbb{F}_{p}$. In the sequel, we write $\eta_{i}^{(4,p)}$ and $C_i^{(4,p)}$ as $\eta_{i}$ and $C_i$, respectively, until stated. The following lemma plays an important role in determining the complete weight enumerator, in which the value of $\eta_0$ coincides with the result of Theorem 4.2.4 of~\cite{berndt1998gauss}.

%

\begin{lemma}\label{value of gauss period}
Let $p\equiv1\mod 4$. Let $s$ and $t$ be defined by $p=s^2+t^2$, $s\equiv1\mod 4$. The Gauss periods of order 4 over $\mathbb{F}_{p}$ are given as follows.\\
$(i)$ If $p\equiv5\mod 8$, then
\begin{eqnarray*}
\{\eta_0,\eta_2\}&=&\left\{\frac{\sqrt{p}-1}{4}\pm\frac{\sqrt{2}}{4}\sqrt{-\sqrt{p}s-p}\right\},\\
\{\eta_1,\eta_3\}&=&\left\{-\frac{\sqrt{p}+1}{4}\pm\frac{\sqrt{2}}{4}\sqrt{\sqrt{p}s-p}\right\}.
\end{eqnarray*}
\\
$(ii)$ If $p\equiv1\mod 8$, then
\begin{eqnarray*}
\{\eta_0,\eta_2\}&=&\left\{\frac{\sqrt{p}-1}{4}\pm\frac{\sqrt{2}}{4}\sqrt{p-\sqrt{p}s}\right\},\\
\{\eta_1,\eta_3\}&=&\left\{-\frac{\sqrt{p}+1}{4}\pm\frac{\sqrt{2}}{4}\sqrt{p+\sqrt{p}s}\right\}.
\end{eqnarray*}
\end{lemma}

\begin{proof}

According to~\cite{myerson1981period}, the Gauss sums $G_i$ are given by
\begin{eqnarray*}
G_i=\sum_{x\in \mathbb{F}_{p}} \zeta_p^{\beta^i x^4},~~i=0,1,2,3,
\end{eqnarray*}
and they are roots of a polynomial $F_4(X)$, i.e.,
\begin{eqnarray*}
F_4(X)=\prod_{i_0}^3 (X-G_i),
\end{eqnarray*}
which is called reduced (or modified) period polynomial.
By Theorem 14 of~\cite{myerson1981period} (see also Theorem 10.10.6 of~\cite{berndt1998gauss}), we have
\begin{eqnarray*}
F_4(X)=\left\{\begin{array}{lll}
(X^2+3p)^2-4p(X-s)^2 ~~~~ && \mathrm{if} ~~p\equiv5\mod 8,\\
(X^2-p)^2-4p(X-s)^2 ~~~~  &&\mathrm{if} ~~p\equiv1\mod 8,
\end{array} \right.
\end{eqnarray*}
where $p=s^2+t^2$ with $s\equiv1\mod 4$.

In the following, we give the proof of case $p\equiv5\mod 8$ since that of case $p\equiv1\mod 8$ is similarly verified.

In the case of $p\equiv5\mod 8$, we have
\begin{eqnarray*}
F_4(X)= \left(X^2+3p-2\sqrt{p}(X-s)\right)\left(X^2+3p+2\sqrt{p}(X-s)\right).
\end{eqnarray*}

Note that $\eta_0+\eta_2 = \eta_0^{(2,p)}=\frac{1}{2}(\sqrt{p}-1)$
yields that $G_0+G_2=2\sqrt{p}$, since $G_i=4\eta_i+1$. Hence, we see that $G_0$, $G_2$ are roots of
\begin{eqnarray*}
X^2+3p-2\sqrt{p}(X-s)=0.
\end{eqnarray*}
Therefore, $G_1$, $G_3$ are roots of
\begin{eqnarray*}
X^2+3p+2\sqrt{p}(X-s)=0.
\end{eqnarray*}
It is straightforward that
\begin{eqnarray*}
&G_0+G_2=~~2\sqrt{p},~~& G_0G_2=3p+2\sqrt{p} s,\\
&G_1+G_2=-2\sqrt{p},~~& G_1G_3=3p-2\sqrt{p} s.
\end{eqnarray*}
Moreover, we can obtain that
\begin{eqnarray*}
&&\eta_0\eta_2= \frac{1}{16}(3p+1-2\sqrt{p}(1-s)),\\
&&\eta_1\eta_3= \frac{1}{16}(3p+1+2\sqrt{p}(1-s)),\\
&&\eta_0^2+\eta_2^2= \frac{1}{8}(1-p-2\sqrt{p}(1+s)),\\
&&\eta_1^2+\eta_3^2= \frac{1}{8}(1-p+2\sqrt{p}(1+s)).
\end{eqnarray*}
Consequently, we have
\begin{eqnarray*}
(\eta_0+\eta_2)^2 &= \frac{1}{4}(\sqrt{p}-1)^2,~~
(\eta_0-\eta_2)^2 &= \frac{1}{2}(-\sqrt{p}s-p),\\
(\eta_1+\eta_3)^2 &= \frac{1}{4}(\sqrt{p}+1)^2,~~
(\eta_1-\eta_3)^2 &= \frac{1}{2}(\sqrt{p}s-p).
\end{eqnarray*}

The desired conclusions follow from the facts that
$\eta_0+\eta_2 =\frac{1}{2}(\sqrt{p}-1)$ and
$\eta_0+\eta_1+\eta_2+\eta_3=-1$.
\hfill\space$\qed$
\end{proof}

 \subsection{The proof of Theorem~\ref{thcwe:CD1,-1}}\label{subsec:pf2}

Observe that $a=0$ gives the zero codeword and the contribution to the complete
weight enumerator is $w_0^{n}$, where $n=\frac{p-1}{2}p^{m-1}$. This value occurs only once. Hence, we assume that $a\in\mathbb{F}_{r}^*$ for the rest of the proof.

For $\rho\in \mathbb{F}_{p}^*$, we consider
\begin{eqnarray*}
   A=\sum_{x\in\mathbb{F}_{r}}\sum_{y\in\mathbb{F}_{p}}\zeta_p^{y^2\mathrm{Tr}(x)}
     \sum_{z\in\mathbb{F}_{p}}\zeta_p^{z\mathrm{Tr}(ax^2)-z\rho}.
 \end{eqnarray*}
Then, it is easy to see that
\begin{eqnarray}\label{eq:A}
  A = N_0(\rho)p^2+(N(\rho)-N_1(\rho))p\sqrt{p^*},
\end{eqnarray}
since
\begin{eqnarray*}
  ~~~~~~~~~~~\sum_{y\in\mathbb{F}_{p}}\zeta_p^{y^2\mathrm{Tr}(x)}
  =\left\{\begin{array}{lll}
    p     &&~~\mathrm{if}~~\mathrm{Tr}(x)=0,\\
    \sqrt{p^*}                        &&~~\mathrm{if}~~\mathrm{Tr}(x)\in Sq,\\
    -\sqrt{p^*}    &&~~\mathrm{if}~~\mathrm{Tr}(x)\in Nsq,
\end{array} \right.
 \end{eqnarray*}
and
\begin{eqnarray*}
  \sum_{z\in\mathbb{F}_{p}}\zeta_p^{z\mathrm{Tr}(ax^2)-z\rho}
  =\left\{\begin{array}{lll}
    p     &&~~\mathrm{if}~~\mathrm{Tr}(ax^2)=\rho,\\
    0                         &&~~\mathrm{if}~~\mathrm{Tr}(ax^2)\neq \rho.
\end{array} \right.
 \end{eqnarray*}

On the other hand, from Theorem 5.33 of~\cite{lidl1983finite} and Equation~\eqref{eq:ortho}, we get
\begin{eqnarray}\label{eq:A2}
  A &=&r+\sum_{y\in\mathbb{F}_{p}^*}\sum_{x\in\mathbb{F}_{r}}\zeta_p^{y^2\mathrm{Tr}(x)}
       +\sum_{z\in\mathbb{F}_{p}^*}\zeta_p^{-z\rho}
        \sum_{y\in\mathbb{F}_{p}}\sum_{x\in\mathbb{F}_{r}}\zeta_p^{\mathrm{Tr}(azx^2+y^2x)}\nonumber \\
    &=&r+\sum_{z\in\mathbb{F}_{p}^*}\zeta_p^{-z\rho}
        \sum_{y\in\mathbb{F}_{p}}\zeta_p^{\mathrm{Tr}(-\frac{y^4}{4az})}\eta(az)G(\eta)\nonumber  \\
    &=&r+\eta(a)G(\eta)\sum_{z\in\mathbb{F}_{p}^*}\zeta_p^{-z\rho}\eta(z)
        \sum_{y\in\mathbb{F}_{p}}\zeta_p^{-\frac{\mathrm{Tr}(a^{-1})}{4z}y^4}.
 \end{eqnarray}

In the following, we calculate the value $A$ of~\eqref{eq:A2} by distinguishing the cases of $\mathrm{Tr}(a^{-1})=0$ and $\mathrm{Tr}(a^{-1})\neq0$.

$\emph{Case 1}$: $\mathrm{Tr}(a^{-1})=0$.

In this case, we have
\begin{eqnarray*}
  A&=&\left\{\begin{array}{lll}
   r-p\eta(a)G(\eta)  && \mathrm{if}~~m~~\mathrm{even},\\
   r+p\eta(a)\bar{\eta}(-\rho)G(\eta)G(\bar{\eta}) && \mathrm{if}~~m~~\mathrm{odd},
\end{array} \right.
 \end{eqnarray*}
which leads to $N(\rho)=N_1(\rho)$ compared with Equation~\eqref{eq:A} and Lemma~\ref{lem:N0 rho}.
It follows from Lemma~\ref{lem:sum2} that
$N(\rho)=\frac{p-1}{2}p^{m-2}$.
This value occurs $p^{m-1}-1$ times.

$\emph{Case 2}$: $\mathrm{Tr}(a^{-1})\neq0$.

Recall that $p\equiv3 \mod 4$. Thus, $\mathrm{gcd}(4,p-1)=2$.
From Equation~\eqref{eq:A2} and Lemma~\ref{lm:expo sum k}, we have
\begin{eqnarray*}
  A&=&r+\eta(a)G(\eta)\sum_{z\in\mathbb{F}_{p}^*}\zeta_p^{-z\rho}\eta(z)
        \bar{\eta}\left(-\frac{\mathrm{Tr}(a^{-1})}{4z}\right)G(\bar{\eta})\\
   &=&r+\eta(a)G(\eta)\bar{\eta}(-\mathrm{Tr}(a^{-1}))\sum_{z\in\mathbb{F}_{p}^*}\zeta_p^{-z\rho}\eta(z)
        \bar{\eta}(z)G(\bar{\eta})\\
   &=&\left\{\begin{array}{lll}
   r+\eta(a)\bar{\eta}(\rho\mathrm{Tr}(a^{-1}))G(\eta)G(\bar{\eta})^2
                      &&~~ \mathrm{if}~~m~~\mathrm{even},\\
   r-\eta(a)\bar{\eta}(-\mathrm{Tr}(a^{-1}))G(\eta)G(\bar{\eta})
                      &&~~ \mathrm{if}~~m~~\mathrm{odd},
\end{array} \right.
 \end{eqnarray*}
which also leads to $N(\rho)=N_1(\rho)$ from Equation~\eqref{eq:A} and Lemma~\ref{lem:N0 rho}.
It then follows from Lemma~\ref{lem:sum2} that
\begin{eqnarray*}
&&N(\rho)=\left\{\begin{array}{lll}
   \frac{p-1}{2}p^{m-2}             && ~~\mathrm{if}~~\bar{\eta}(\rho\mathrm{Tr}(a^{-1}))=1\\
   \frac{p-1}{2}p^{m-2}+\eta(a)(-1)^{\frac{m}{2}}p^{\frac{m-2}{2}}
     && ~~\mathrm{if}~~\bar{\eta}(\rho\mathrm{Tr}(a^{-1}))=-1
\end{array} \right.
 \end{eqnarray*}
for even $m$, and otherwise,
\begin{eqnarray*}
&&N(\rho)= \frac{p-1}{2}p^{m-2}+\frac{1}{2}\eta(a)(-1)^{\frac{m-1}{2}}p^{\frac{m-3}{2}}
(p\bar{\eta}(\rho)+\bar{\eta}(\mathrm{Tr}(a^{-1}))).
\end{eqnarray*}

Note that $N(0)=\frac{p-1}{2}p^{m-1}-\sum_{\rho\in \mathbb{F}_{p}^* }N(\rho)$.
The desired conclusion then follows from Lemmas~\ref{lem:nij1} and~\ref{lem:nij2}.

This completes the proof of Theorem~\ref{thcwe:CD1,-1}.

 \subsection{The proof of Theorem~\ref{thcwe:CD1,-1 p=1mod4}}\label{subsec:pf3}

By the proof of Theorem~\ref{thcwe:CD1,-1}, we only need to consider the case $\mathrm{Tr}(a^{-1})\neq0$ with $a\in\mathbb{F}_{r}^*$, since the cases of $a=0$ and $\mathrm{Tr}(a^{-1})=0$ have already been determined. For this purpose, we write Equation~\eqref{eq:A2} as
\begin{eqnarray}\label{eq:A4}
  A =r+\eta(a)G(\eta)B,
 \end{eqnarray}
 where
\begin{eqnarray}\label{def:B}
B=\sum_{z\in\mathbb{F}_{p}^*}\zeta_p^{-z\rho}\eta(z)
        \sum_{y\in\mathbb{F}_{p}}\zeta_p^{-\frac{\mathrm{Tr}(a^{-1})}{4z}y^4}.
\end{eqnarray}

Let notations be as aforementioned and $p\equiv1\mod 4$. When $\mathrm{Tr}(a^{-1})\neq0$, the value of $B$
can be determined by
\begin{eqnarray}\label{valueofB}
B &=& \sum_{z\in\mathbb{F}_{p}^*}\zeta_p^{-z\rho}\bar{\eta}(z)
       \left(4\eta_{-\frac{\mathrm{Tr}(a^{-1})}{4z}}+1\right) \nonumber\\
  &=& \left(\sum_{z\in C_0}+\sum_{z\in C_2}-\sum_{z\in C_1}-\sum_{z\in C_3} \right) 4\zeta_p^{-z\rho}\eta_{-\frac{\mathrm{Tr}(a^{-1})}{4z}} + \bar{\eta}(-\rho)G(\bar{\eta}) \nonumber\\
  &=& \left(\sum_{z\in C_0}+\sum_{z\in C_2}-\sum_{z\in C_1}-\sum_{z\in C_3} \right) 4\zeta_p^{-z\rho}\eta_{-\frac{\mathrm{Tr}(a^{-1})}{4z}} + \bar{\eta}(\rho)\sqrt{p},
\end{eqnarray}
since $m$ is odd. By Equations~\eqref{eq:A},~\eqref{eq:A4}, and Lemma~\ref{lem:sum2}, we have
\begin{eqnarray}\label{eqN1andN-1}
\left\{\begin{array}{lll}
N(\rho)+N_1(\rho)&=& p^{m-1}-p^{m-2}+\eta(a)p^{\frac{m-3}{2}}\left(\bar{\eta}(\rho) p+\bar{\eta}(\mathrm{Tr}(a^{-1}))\right),\\
N(\rho)-N_1(\rho)&=& \eta(a)\left(\bar{\eta}(\mathrm{Tr}(a^{-1}))p^{\frac{m-2}{2}}+p^{\frac{m-3}{2}}B\right).
\end{array} \right.
\end{eqnarray}

Now, we assume that $p\equiv5\mod 8$.

Clearly, $-1$ and $4$ are both in $C_2$. In the following, the value $B$ of \eqref{valueofB} will be computed according to the choices of $\mathrm{Tr}(a^{-1})$ and $\rho$.

$\emph{Case 1:}$ $\mathrm{Tr}(a^{-1})\in C_0$, $\rho\in C_0$.

In this case, by Lemma~\ref{value of gauss period} and Equation~\eqref{valueofB}, we obtain
\begin{eqnarray*}
B = 4(2\eta_0\eta_2-\eta_1^2-\eta_3^2)+\sqrt{p}= 2p-\sqrt{p}.
\end{eqnarray*}

It follows from Equation~\eqref{eqN1andN-1} that
\begin{eqnarray*}
&N(\rho)&= \frac{p-1}{2}p^{m-2}+\frac{1}{2}\eta(a)
(3 p^{\frac{m-1}{2}}+p^{\frac{m-3}{2}}),\\
&N_1(\rho)&= \frac{p-1}{2}p^{m-2}-\frac{1}{2}\eta(a)
(~ p^{\frac{m-1}{2}}-p^{\frac{m-3}{2}}).
\end{eqnarray*}

$\emph{Case 2:}$ $\mathrm{Tr}(a^{-1})\in C_0$, $\rho\in C_1$.

In this case, we deduce that
\begin{eqnarray*}
B =4(\eta_3\eta_0+\eta_1\eta_2-\eta_0\eta_3-\eta_2\eta_1)-\sqrt{p}=-\sqrt{p},
 \end{eqnarray*}
which indicates that
\begin{eqnarray*}
N(\rho)=N_1(\rho)= \frac{p-1}{2}p^{m-2}-\frac{1}{2}\eta(a)
(~ p^{\frac{m-1}{2}}-p^{\frac{m-3}{2}}).
\end{eqnarray*}

$\emph{Case 3:}$ $\mathrm{Tr}(a^{-1})\in C_0$, $\rho\in C_2$.

In this case, we have
\begin{eqnarray*}
B =4(\eta_0^2+\eta_2^2-2\eta_1\eta_3)+\sqrt{p}=-2p-\sqrt{p},
 \end{eqnarray*}
which gives that
\begin{eqnarray*}
&N(\rho)&= \frac{p-1}{2}p^{m-2}-\frac{1}{2}\eta(a)
(~ p^{\frac{m-1}{2}}-p^{\frac{m-3}{2}}),\\
&N_1(\rho)&= \frac{p-1}{2}p^{m-2}+\frac{1}{2}\eta(a)
(3 p^{\frac{m-1}{2}}+p^{\frac{m-3}{2}}).
\end{eqnarray*}

$\emph{Case 4:}$ $\mathrm{Tr}(a^{-1})\in C_0$, $\rho\in C_3$.

In this case, we obtain
\begin{eqnarray*}
B =4(\eta_1\eta_0+\eta_3\eta_2-\eta_2\eta_3-\eta_0\eta_1)-\sqrt{p}=-\sqrt{p}.
 \end{eqnarray*}
As a consequence, we get
\begin{eqnarray*}
N(\rho)=N_1(\rho)= \frac{p-1}{2}p^{m-2}-\frac{1}{2}\eta(a)
(~ p^{\frac{m-1}{2}}-p^{\frac{m-3}{2}}).
\end{eqnarray*}

Moreover, for $\mathrm{Tr}(a^{-1})\in C_0$, the number of $a$ satisfying
$\eta(a)=1$ is
\begin{eqnarray*}
\#\{a\in \mathbb{F}_{r}^*:\eta(a)=1,~ \mathrm{Tr}(a^{-1})\in C_0 \}
=\frac{1}{2}n_{1,1}
=\frac{p-1}{8}(p^{m-1}+ p^{\frac{m-1}{2}}),
\end{eqnarray*} according to Lemma~\ref{lem:nij1}.
And similarly, the number of $a$ satisfying
$\eta(a)=-1$ is
\begin{eqnarray*}
&&\#\{a\in \mathbb{F}_{r}^*:\eta(a)=-1,~ \mathrm{Tr}(a^{-1})\in C_0 \}
=\frac{1}{2}n_{-1,1}
=\frac{p-1}{8}(p^{m-1}- p^{\frac{m-1}{2}}),
\end{eqnarray*}
according to Lemma~\ref{lem:nij2}.

There are sixteen cases all together to be considered. Other cases can be similarly calculated, which are omitted here.

Note that the case of $p\equiv1\mod 8$ can be analyzed
in an analogous fashion. The proof of Theorem~\ref{thcwe:CD1,-1 p=1mod4} is finished.

 \subsection{The proof of Theorem~\ref{thcwe:CD1,-1 p=1mod4 2}}\label{subsec:pf4}

This proof is similar to that of Theorem~\ref{thcwe:CD1,-1 p=1mod4} by observing that
\begin{eqnarray*}\label{valueofB2}
B &=&  \left(\sum_{z\in C_0}+\sum_{z\in C_1}+\sum_{z\in C_2}+\sum_{z\in C_3} \right) 4\zeta_p^{-z\rho}\eta_{-\frac{\mathrm{Tr}(a^{-1})}{4z}} -1,
\end{eqnarray*}
from Equation~\eqref{def:B}, since $m$ is even.
Thus, we omit the details here.

\section{Concluding remarks}\label{sec:conclusion}

Inspired by the original ideas of~\cite{ding2015twothree,tang2015linear}, we constructed a class of three-weight linear codes. By employing some mathematical tools, we presented explicitly their complete weight enumerators and weight enumerators. Their punctured codes contain some almost optimal codes. By Theorem~\ref{wt:CD1,-1}, it is easy to check that
\begin{eqnarray*}
\frac{w_{min}}{w_{max}}>\frac{p-1}{p},
\end{eqnarray*}
for $m\geq 4$. Here $w_{min}$ and $w_{max}$ denote the minimum and maximum nonzero weights in $C_{D}$, respectively. Therefore, the code $C_{D}$ can be used for secret sharing schemes with interesting access structures. We also mention that the complete weight enumerators, presented in Theorems~\ref{thcwe:CD1,-1},~\ref{thcwe:CD1,-1 p=1mod4} and~\ref{thcwe:CD1,-1 p=1mod4 2}, can be applied to construct systematic authentication codes. Furthermore, if $r$ is large enough, these authentication codes are asymptotically optimal. See~\cite{Ding2005auth,ding2015twothree,LiYang2015cwe}.

 Note that $\mathrm{gcd}(4,p-1)=4$ if $p\equiv1\mod 4$. This implies that we can prove Theorems~\ref{thcwe:CD1,-1 p=1mod4} and~\ref{thcwe:CD1,-1 p=1mod4 2} with a similar method used in
 Subsection~\ref{subsec:pf2}. One can see that it works well though it is indeed very complicated. However, we gave a simpler proof by
 employing Gauss periods to determine the complete weight enumerator of $C_{D}$ for the case of $p\equiv1\mod 4$.

To conclude this paper, we remark that the codes proposed in this paper can be extended to a more general case, that is, for an integer $t\geq2$, define
\begin{eqnarray*}\label{def:generalCD1-1}
    C_{D'}=\left\{\left(\mathrm{Tr}(a_1x_1^2+\cdots+a_tx_t^2)\right)_{(x_1,\cdots,x_t)\in D}: a_1,\cdots,a_t\in \mathbb{F}_{r}\right\},
\end{eqnarray*}
where
\begin{eqnarray*}
     D'   =\left\{(x_1,\cdots,x_t)\in \mathbb{F}_{r}^t:\mathrm{Tr}(x_1+\cdots+x_t)\in Sq\right\}.
\end{eqnarray*}
For this kind of linear codes, it will be an interesting work to settle their complete weight enumerators.

\begin{acknowledgements}
The work of Zheng-An Yao is partially supported by the NSFC (Grant No.11271381), the NSFC (Grant No.11431015)
and China 973 Program (Grant No. 2011CB808000).
The work of Chang-An Zhao is partially supported by the NSFC (Grant No. 61472457). This work is also partially supported by Guangdong Natural Science
Foundation (Grant No. 2014A030313161).
\end{acknowledgements}


\end{document}